\documentclass[prl,twocolumn,aps,preprintnumbers,showpacs,superscriptaddress]{revtex4-1}
\usepackage{amsfonts}
\usepackage{amsmath}
\usepackage{pspicture}
\usepackage{epsfig}
\usepackage{calc}
\usepackage{subfigure} 
\usepackage{graphicx} 
\usepackage{enumerate}  
\usepackage{amssymb}
\usepackage{amsthm}
\newcommand{\appref}[1]{\hyperref[#1]{{Appendix~\ref*{#1}}}}
\newcommand{\be}{\begin{eqnarray} \begin{aligned}}
\newcommand{\ee}{\end{aligned} \end{eqnarray} }
\newcommand{\benn}{\begin{eqnarray*} \begin{aligned}}
\newcommand{\eenn}{\end{aligned} \end{eqnarray*}}
\newcommand{\cancel}[1]{} 
\newcommand*{\cA}{\mathcal{A}} 
\newcommand*{\cB}{\mathcal{B}}
\newcommand*{\cC}{\mathcal{C}}
\newcommand*{\cE}{\mathcal{E}}

\newcommand*{\cL}{\mathcal{L}}
\newcommand*{\cM}{\mathcal{M}}

\newcommand*{\cD}{\mathcal{D}}

\newcommand*{\cS}{\mathcal{S}}

\newcommand*{\cT}{\mathcal{T}}

\newcommand*{\tr}{\mathop{\mathrm{tr}}\nolimits}

\newcommand*{\supp}{\mathrm{supp}}

\newcommand{\bc}{\begin{center}}
\newcommand{\ec}{\end{center}}

\newcommand{\id}{\mathbb{I}}

\newtheorem{theorem}{Theorem}[]  
\newtheorem{lemma}{Lemma}[]
\newtheorem{definition}{Definition}

\newtheorem{corollary}{Corollary}[theorem]

\usepackage{amsfonts}

\def\id{\mathbb{I}}

\def\01{\{0,1\}}

\newcommand{\ket}[1]{|#1\rangle}

\newcommand{\proj}[1]{|#1\rangle\langle#1|}

\newcommand{\comment}[1]{}

\renewcommand*{\kappa}{k}

\begin{document}
\newcommand*{\spec}{\mathsf{spec}}
\newcommand*{\sspec}{\mathsf{Sspec}}
\newcommand*{\photonnumber}{{\bf N}}
\newcommand*{\Mat}{\mathsf{Mat}}
\newcommand*{\poi}{\mathsf{Poi}}
\newcommand*{\bin}{\mathsf{Bin}}

\newcommand*{\bT}{\overline{\mathcal{T}}}
\renewcommand*{\id}{\mathsf{id}}
\renewcommand*{\P}{{\bf P}}
\newcommand*{\hit}{{\bf t}}
\newcommand*{\cCun}{\mathcal{C}_*}
\newcommand*{\logicalL}{\overline{\mathsf{L}}}

\newcommand*{\Smatrix}[2]{\raisebox{-1.7ex}{
\begin{picture}(39,28)(-2,-4)
\put(0,0){\includegraphics[scale=.5]{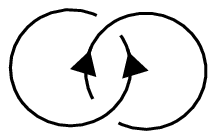}}
\put(-4,5){\small $#1$}
\put(30,5){\small $#2$}
\end{picture}}}

\title{Generating topological order: no speedup by dissipation}
\author{Robert \surname{K\"onig}}
\affiliation{Institute for Quantum Computing and Department of Applied Mathematics, University of Waterloo
}

\author{Fernando \surname{Pastawski}}

\affiliation{Institute for Quantum Information and Matter, California Institute of Technology
}
\begin{abstract}
We consider the problem of preparing  topologically ordered states
using unitary and non-unitary circuits, as well as local time-dependent Hamiltonian  and Liouvillian evolutions. We prove that for any topological code in $D$~dimensions, the time required to encode logical information into the ground space is at least~$\Omega(d^{1/(D-1)})$, where~$d$ is the code distance. 
This result is tight for the toric code, giving a scaling with the linear system size. 
More generally, we show that the linear scaling is necessary even when dropping the requirement of encoding:  preparing any state close to the ground space using dissipation takes an amount of time proportional to the diameter of the system in typical $2D$~topologically ordered systems, as well as for example the 3D and 4D toric codes.
\end{abstract}
\maketitle

\section{Introduction}
Topological codes, such as Kitaev's toric code or his quantum double models~\cite{Kitaevtoric}, the Levin-Wen model~\cite{LevinWen05}, or Bombin and Martin-Delgado's color codes~\cite{BombinDelgado07}, are a potential platform for the realization of robust quantum computation. 
Such a code is associated with a many-body  system
of qudits  arranged on the vertices of a regular lattice~$\Lambda$ in~$D$~spatial dimensions. Remarkably,  syndrome information can be extracted by measuring local  observables. Furthermore, the code distance is typically macroscopic, i.e., scales with the system size. These features promise to greatly facilitate the fault-tolerant storage and manipulation of  quantum information. 

Here we consider the problem of encoding states into such a code. This is the problem of transferring local (unprotected) information on a few qudits into the many-qudit ground space. More generally, we consider the generation of topological order: here we do not ask to prepare any specific state, but rather an arbitrary (possibly unknown) ground state. We show that both these problems are hard if no global control or interactions can be used. This is true irrespective of whether or not we allow dissipative processes or time-dependent interactions: any geometrically local encoding map takes a time scaling with code distance in~$2D$. Similarly, generation of typical $2D$~topological order on an $L\times L$~lattice takes at least a linear amount of time in~$L$. These bounds are tight for the toric code~\cite{DenKoePas13} in~$2D$. Our bounds  also yield statements for higher dimensions although we do not believe those are tight.

 To jointly treat the mentioned  (and potentially yet undiscovered) examples of topological codes, we formulate our results in the general framework of commuting  projector Hamiltonians. That is, the code space~$\cC\subset (\mathbb{C}^p)^{\otimes n}$ is the simultaneous $+1$-eigenspace of a family~$\{\Pi_a\}_a$ of pairwise commuting projections or, equivalently, the ground space of the Hamiltonian $H=-\sum_a \Pi_a$. The features distinguishing the class of  topological codes are
\begin{enumerate}[(i)]
\item the locality of the projections~$\{\Pi_a\}_a$: the support of each projection has diameter (on the lattice) upper bounded by the same constant~$\xi$.
 The support of an operator~$\Pi$ is the set of  qudits on which it acts non-trivially. 
\item
the code distance is an extensive function of the system size, typically e.g., of the form $d=\Omega(n^\alpha)$ for some constant $\alpha>0$.
\end{enumerate}
We  shall call a code~$\cC$ satisfying these properties a topological code. Our first result, Theorem~\ref{thm:generallowerbound} below applies to all such codes, and characterizes  encoders based on locality-preserving evolutions as defined below. Our second result, Theorem~\ref{thm:preparationtoporder}, bounds  the potential of locality-preserving evolutions to generate any (possibly undetermined) ground states. It applies to a subclass of topological codes described in detail below. This subclass includes all anyonic models in~$2D$, as well as higher-dimensional toric codes.

\subsection{Encoders for the code}
Given a code~$\cC$, we are interested in  encoders, i.e., completely positive trace-preserving maps (CPTPM)~$\cE$ which convert `simple' states into code states. We will use $A=A_1\cdots A_\kappa\subset\Lambda$ to refer to the $\kappa$-qudits that are being encoded. Here we are assuming that the code space has dimension~$p^\kappa$ for simplicity -- the generalization is straightforward.
Since our encoder~$\cE$ is supposed to take locally encoded information into the code space, we  assume that the qudits in~$A$ 
are nearest neighbors (i.e., form a simply connected subset of the lattice). Independently of the logical information that is being encoded, we assume that the remaining $n-\kappa $ qudits~$A^c=\Lambda\backslash A$ are in
a fixed product state~$\ket{\Phi}=\bigotimes_{j=1}^{n-\kappa}\ket{\Phi_j}\in(\mathbb{C}^p)^{\otimes( n-\kappa) }$. This is intended to be a state which is easy to prepare (see~\cite{DenKoePas13}). In fact, neither the product form nor the fact that it is pure are essential for our lower bound on encoding time. We use this convention as it corresponds to a natural operational restriction.

Given this setup, the notion of a (perfect) encoder is particularly easy to define for unitary maps~$U:(\mathbb{C}^p)^{\otimes n}\rightarrow (\mathbb{C}^p)^{\otimes n}$. A unitary encoder~$U$ takes the subspace
\begin{align}
\cCun:=(\mathbb{C}^p)^{\otimes \kappa}\otimes \mathbb{C}\ket{\Phi}\subset (\mathbb{C}^\kappa)^{\otimes n}\label{eq:ccun}
\end{align}
isomorphically to~$\cC$.  For general physical maps, that is, completely positive trace-preserving maps (CPTPMs), we define (approximate) encoders similarly as follows. 
\begin{definition}\label{def:dissipativeencoder}
A CPTPM $\cE$ {\em encodes $\cCun$ (cf.~\eqref{eq:ccun}) into the code~$\cC$ with error~$\epsilon$} if
\begin{align}
\|\cE(\cdot)-U\cdot U^\dagger\|_{\cS(\cCun)}\leq \epsilon\ .\label{eq:semigroupdistance}
\end{align}
Here $U$ is an (arbitrary) unitary encoder for~$\cC$, and  we use the norm
\begin{align}
\|\cE(\cdot)\|_{\cS(\cCun)}:=\max_{\rho: \substack{\rho\geq 0, \\
\supp(\rho)\subset \cCun\\
\tr(\rho)=1}}\|\cE(\rho)\|_1\ \label{eq:unstablenorm}
\end{align}
obtained by maximizing over all states with support in~$\cC_*$.
\end{definition}
\noindent Note that using the given notion of  distance in Definition~\ref{def:dissipativeencoder} (instead of the diamond norm) strengthens our lower bound on the encoding time. 

In a similar manner, we can define the notion of a preparation map:
\begin{definition}
A CPTPM~$\cE$ prepares a state in~$\cC$ with error~$\epsilon$ if
there is a product state~$\ket{\Phi}=\bigotimes_{j=1}^n\ket{\Phi_j}\in(\mathbb{C}^p)^{\otimes n}$ such that 
\begin{align*}
\min_{\substack{
\rho\geq 0\\
\supp(\rho)\subset\cC\\
\tr(\rho)=1}}\|\cE(\proj{\Phi})-\rho\|_1\leq \epsilon
\end{align*}
\end{definition}
Again, the fact that the initial state~$\ket{\Phi}$ is pure is unimportant and could be omitted. The main property of the state  we will need is that it has no classical correlations among distant qudits.

\subsection{Locality-preserving evolutions}
Our results apply to arbitrary evolutions which preserve locality. To define this notion in more detail, let us say that a CPTPM~$\cE$ is localizable with error~$G(r)$ 
iff for any $r>0$ and region~$B\subset\Lambda$, there exists a
 CPTPM~$\cE_{B(r)}$ supported on the $r$-neighborhood
$B(r):=\{x\in \Lambda\ |\ d(x,B)\leq r\}$ of $B$ such that
\begin{align*}
\|\cE^\dagger(O_B)-\cE_{\cB(r)}^\dagger(O_B)\|_\infty \leq \|O_B\|_\infty \cdot |B|\cdot G(r)\ 
\end{align*}
for any observable~$O_B$  supported on~$B$. (Here $\cE^\dagger$ is the adjoint map with respect to the Hilbert-Schmidt inner product.)
We shall call a family~$\{\cE^{(t)}\}$ of evolution operators (corresponding to some `time' $t$) {\em of  Lieb-Robinson-type} if each $\cE^{(t)}$ is localizable with error $G(r)=Ce^{vt-\gamma r}$ for some non-negative constants $C,v,\gamma$. A trivial example is a family $\{\cE^{(t)}\}_{t\geq 0}$ of (unitary or non-unitary) circuits where each~$\cE^{(t)}$ has circuit depth upper bounded by~$t$. Also, as argued in~\cite{Bravyietalpropagationtop} using the Lieb-Robinson bound, any time-dependent local Hamiltonian~$H(t)$  (or Hamiltonian with exponentially decaying interactions)  generates
via $U(t)=T\exp\left[i\int_0^t H(s)ds\right]$  a family $\cE^{(t)}(\cdot)=U(t)\cdot U(t)^\dagger$ of Lieb-Robinson type. Similarly and more generally, any time-dependent local Liouvillian~$\cL(t)$  with  bounded-strength, constant-range (or exponentially decaying) interactions  generates a family 
\begin{align}
\cE^{(t)}=T\exp\left[\int_0^t \cL(s)ds\right]\label{eq:dissipativeencoderdef}
\end{align} of Lieb-Robinson type, see~\cite[Lemma 5.3]{stabilitylocal} or \cite[Theorem 2]{BarthelKliesch}. The corresponding localized evolution~$\cE^{(t)}_{B(r)}$ is generated by the Liouvillian~$\cL_{B(r)}=\sum_{X\subset B(r)}\cL_X$ obtained by neglecting terms with support outside~$B(r)$ in the sum~$\cL=\sum_{X\subset\Lambda} \cL_X$. 

\section{Main results}
Our main result for encoders is the following:
\begin{theorem}\label{thm:generallowerbound}
Let $\cC$ be a $D$-dimensional topological code with distance~$d$. 
 Assume that  $\{\cE^{(t)}\}_t$ is a family of CPTPMs which is of Lieb-Robinson type. Assume further that for some $t>0$, $\cE^{(t)}$ encodes~$\cCun$ into~$\cC$ with constant error $\epsilon\ll 1$.  Then~$t\geq \Omega(d^{1/(D-1)})$.
\end{theorem}
\noindent For the  paradigmatic case of the $2D$~toric code, we obtain:
\begin{corollary}\label{cor:toriccode}
Consider Kitaev's toric code on an $L\times L$ lattice. Then any dissipative encoder takes time at least linear in $L$.
\end{corollary}
In~\cite{DenKoePas13}, an explicit construction of a time-independent Liouvillian is given which  acts as an encoder for the toric code in linear time. This shows that Corollary~\ref{cor:toriccode} is tight. Corollary~\ref{cor:toriccode} genereralizes the result of~\cite{Bravyietalpropagationtop}, where a linear lower bound is shown for unitary encoders. 
Note that the best known geometrically local unitary encoder~\cite{dennisetal02}  for the toric code takes time~$\Theta(L^2)$. Our result hence shows that the speedup of dissipative compared  to unitary processes is at most linear in~$L$.
 
For the problem of generation of topological order, we need to introduce a few additional notions. We shall call an observable~$\logicalL$
 logical if $\logicalL$ commutes with all projections~$\Pi_a$ defining the code, and if it acts non-trivially on the subspace~$\cC$. 
Any measurement (or POVM)~$\cM$ which produces the measurement statistics corresponding to the observable~$\logicalL$ will be called a realization of~$\logicalL$.  If, for a given state~$\rho$, measuring~$\logicalL$ leads to a non-deterministic outcome, we call $\logicalL$ an uncertain observable in state~$\rho$. In fact, we need to be more precise as we are generally dealing with a family of codes parametrized by their system size~$L$. For such a family, a logical observable~$\logicalL$ (with a constant number of eigenvalues) is said to be uncertain for~$\rho$ if the measurement uncertainty, as measured e.g., by the Shannon entropy, is a non-zero constant independent of the system size. 
 
\begin{theorem}\label{thm:preparationtoporder}
Let $\cC$ be the codespace of  a $D$-dimensional topological code parametrized by~$L$. 
Assume that for any state $\rho$ supported on~$\cC$, there is an uncertain logical observable~$\logicalL$ having two realizations $\cM^{(1)},\cM^{(2)}$ satisfying
\begin{align*}
d(\supp(\cM^{(1)}),\supp(\cM^{(2)}))\geq cL\  
\end{align*}
for some constant $c>0$.
Let $\{\cE^{(t)}\}_t$ be a family of CPTPMs  of Lieb-Robinson type. If $ \| \cE^{(t)}(\Phi) -\rho \|_1 \leq \epsilon $  for some initial product state $\Phi$, a constant~$0<\epsilon\ll 1$ and some state $\rho$ supported on~$\cC$, then $t = \Omega(L)$. 
\end{theorem}

We will argue that many natural codes satisfy the conditions of Theorem~\ref{thm:preparationtoporder}. 
This is particularly easy to see in cases where the logical operators are explicitly known. 
For example,  we  show  that preparing a ground state of the  $D$ dimensional toric code on an $L^{\times D}$-lattice takes at least $\Omega(L)$~time. More generally, as we explain below,  Theorem~\ref{thm:preparationtoporder} applies to all topological stabilizer codes originating from translationally invariant generators accomodating a constant number $k$ of encoded qubits (STS codes) if $D\leq 3$. For the important case of topologically ordered systems in~$2D$, we get the following  statement encompassing e.g., the toric code, the Levin-Wen model, or Bombin's color codes:
\begin{corollary}\label{cor:topologicalcodes}
Let~$\cC$ be a topological code on an $L\times L$~periodic lattice of qudits in~$2D$. If~$\cC$ is associated with a topological quantum field theory (TQFT), then preparation of ground states by dissipative processes takes at least linear time in $L$.
\end{corollary}
\noindent The choice of a system on a torus is arbitrary and for concreteness only. However, we require the system to have a ground space degeneracy. Therefore, (similar to earlier work~\cite{Bravyietalpropagationtop}), our results do not provide information about the problem of preparing topological order, e.g., on a sphere.

\section{Non-unitary encoders}
\subsection{Proof sketch for the $2D$ toric code}
Let us briefly recall the argument from~\cite{Bravyietalpropagationtop}, which shows that evolution $U=\cT\exp\left(i\int_0^t H(s)ds\right)$ under a time-dependent Hamiltonian cannot act as an encoder unless the evolution time~$T$ is at least linear in the system size. This is based on the observation that for an arbitrary pair $\bar{\rho}_0,\bar{\rho}_1$ of encoded states, 
we can construct an observable of the form $\logicalL=U (\mathsf{L}\otimes I_{\mathbb{C}^2}^{\otimes n-k})U^\dagger$, which perfectly distinguishes the states.
 Here $\mathsf{L}$ is the (local) $k$-qubit observable that distinguishes the `decoded' states $\rho_0=U^\dagger \bar{\rho}_0 U$ and $\rho_1=U^\dagger\bar{\rho}_1 U$. Lieb-Robinson bounds are then used to argue that~$\logicalL$ is a local operator for small~$t$. Because the code distance is macroscopic (and hence no local operator can distinguish encoded states), this leads to a lower bound on~$t$.  Observe that this argument relies heavily on the fact that the encoder is unitary and its inverse acts as a decoder. 
 
 Conceptually, the previous argument proceeds by constructing a logical operator for the code from a local operator (using the inverse encoder).
 This line of reasoning cannot  simply be translated  to the case of dissipative encoders because these are not invertible. The following proof, while similar in spirit, turns the argument around: logical observables are converted to operators extracting local information from the unencoded states. This circumvents the difficulty of not having an inverse.  However, the relation to the code distance appears to be more subtle in our proof. Indeed, for general topological codes, we require an additional ingredient beyond Lieb-Robinson bounds as we discuss below.

  For didactical reasons, let us first give a rough sketch for the toric code on an $L\times L$ grid. The latter issue does not appear in this case due to the fact that logical operators are known explicitly.  Assume that~$\cE^{(t)}$ encodes $\cCun$ into $\cC$ with small error $\epsilon$, in a time $t\ll L$ (e.g., $t\sim L^{1/2}$). We will show that this leads to a contradiction.
  
  For the toric code, the quantum information (before encoding) is supposed to be localized on two neighboring qubits $A=A_1A_2$.  We begin by choosing a set  $B\subset A^c$ of qubits such that
\begin{enumerate}[(i)]
\item
The distance $d(B,A)=\min_{x\in B} d(x,A)$ between the qubits $B$ and $A$ (on the lattice) is at least $d(B,A)\geq L/3$.
\item
the number of qubits in $B$ is $|B|=L$ and
\item 
There is a non-trivial logical  operator $\bar{P}_B$  with eigenvalues $\{+1,-1\}$ within the code space supported inside $B$.
 \end{enumerate} 
 Concretely, $B$ includes all qubits along cycle on the torus (where the cycle is located at distance $L/3$ from $A$), and $\bar{P}_B=X^{\otimes L}$ can be chosen as the tensor product of the Pauli-$X$ operator along this line.
 
 Next we pick two orthogonal encoded states~$\bar{\rho}_0,\bar{\rho}_1$ which can be perfectly distinguished using the observable~$\bar{P}_B$. Because $\cE^{(t)}$ is an (approximate) encoder by assumption, there are two-qubit states $\rho_0,\rho_1$ such that
 \begin{align*}
\cE^{(t)}(\rho_0\otimes\Phi_{A^c})\approx\bar{\rho}_0\textrm{ and }\cE^{(t)}(\rho_1\otimes\Phi_{A^c})\approx\bar{\rho}_1\ .
 \end{align*}
 Using the fact that $\bar{P}_B$ distinguishes the encoded states on the rhs.~perfectly, we conclude that
 the observable~$(\cE^{(t)})^\dagger(\bar{P}_B)$ distinguishes the unencoded states
 \begin{align}
 \rho_0\otimes\Phi_{A^c}\qquad\textrm{ and }\qquad \rho_1\otimes\Phi_{A^c}\label{eq:complementary}
 \end{align}
 almost perfectly. But according to the Lieb-Robinson bound and the fact that $A$ and $B$ are far from each other, $(\cE^{(t)})^\dagger(\bar{P}_B)$ is an operator with no support on~$A$ for small times.  This contradicts the fact that no such operator can distinguish the two states~\eqref{eq:complementary}.

 \subsection{General proof of Theorem~\ref{thm:generallowerbound}}
 To generalize the proof to arbitrary topological codes, the main additional step  is to show that there exists a region~$B$ with properties analogous to~$(i)-(iii)$.  To do so, we will use the notion of correctable regions introduced in~\cite{BravyiTerhal09} and corresponding `cleaning' results of~\cite{BravyiPoulinTerhal10}. This more general proof also clarifies how the macroscopic code distance comes into play.
 
 Recall that a subset of qudits $\Gamma\subset\Lambda$ is called correctable if the encoded information can be recovered
 even after losing all qudits in~$\Gamma$, i.e., if
 there is a decoding CPTPM~$\cD$ such that
 \begin{align*}
 \cD\circ\tr_{\Gamma}(\bar{\rho})=\bar{\rho}\qquad\textrm{ for all encoded states }\bar{\rho}\in \cB(\cC)\ .
 \end{align*}
 A simple example is any set~$\Gamma$ containing fewer qudits than the code distance, that is, $|\Gamma|<d$. A much less trivial statement which is shown and used in~\cite{BravyiPoulinTerhal10} is the following.
 Let us define the cube $\Gamma_R(v)\subset\Gamma$ for $v\in\Gamma$
as  the  rectangular block of
 size $\underbrace{R\times R\times\cdots\times R}_{D\textrm{ times}}$, i.e., a cube of linear size~$R$ aligned with the coordinate axes and centered around some location~$v$ of the lattice (according to some convention).
 \begin{lemma}[see~\cite{BravyiPoulinTerhal10}]\label{lem:correctableregionlemma}
 Let~$\cC$ be a $D$-dimensional topological code with distance~$d$ and interaction length~$\xi$. Then there is a constant $c=c(\xi)>0$ such that
 all  cubes $\Gamma_{R}(v)$, $v\in\Lambda$ with  $R\leq c d^{\frac{1}{D-1}}$ are correctable.
  \end{lemma}
 Correctable regions are a convenient proof tool because of the following  statement: If $\Gamma\subset\Lambda$ is correctable, and $\bar{P}$ is a logical operator of the code, then there exists a logical operator $\bar{P}'$ with support outside $\Gamma$, such that the actions of $\bar{P}$ and $\bar{P}'$ on the code space~$\cC$ agree. In other words, the support of a logical operator~$\bar{P}$ can be modified to exclude qudits in~$\Gamma$ (explaining the terminology `cleaning') without affecting its action on encoded states.  We will sometimes refer to this as the cleaning lemma~\cite{BravyiTerhal09}. It is an immediate consequence of the definitions.
 
\begin{figure}
\epsfig{file=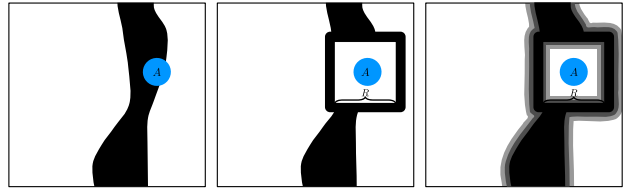,width=\columnwidth}
\caption{This figure illustrates the proof of Theorem~\ref{thm:generallowerbound}. Region $A$ consists of the qubits carrying the logical information before the encoding. A logical operator may have support on this region, but can be cleaned out. Applying the Heisenberg evolution of the encoding map to the operator smears out its support. If the evolution is sufficiently local, the resulting operator still has no support in~$A$ and can therefore not distinguish unencoded states.
  \label{fig:cleaninglemmafigure}}
\end{figure}

With this, we are equipped to give the full proof of Theorem~\ref{thm:generallowerbound}. 
 \begin{proof}
 As before, let $A\subset\Lambda$, $|A|=\kappa$ be the set of qudits carrying the quantum information to be encoded.
  Our first goal is to show the existence of a suitable set~$B$. 
 Let $c$ be the constant from Lemma~\ref{lem:correctableregionlemma}, that is, any cube~$\Gamma_R$ of linear size $R=cd^{\frac{1}{D-1}}$ is correctable see Fig.~\ref{fig:cleaninglemmafigure}. We choose a block~$\Gamma_{R}$ that contains the qudits $A$ near its center. Then we set $B:=\Gamma_{R}^c=\Lambda\backslash\Gamma_R$.  Since $\kappa\ll d$,
  this guarantees that
 \begin{enumerate}[(i)]
 \item
 $d(B,A)\geq \frac{c}{3}d^{\frac{1}{D-1}}$.
 \end{enumerate} 
  Clearly, we also have the trivial bound
 \begin{enumerate}[(i)]\setcounter{enumi}{1}
 \item
 $|B|\leq n$\label{it:upperboundbsize}
 \end{enumerate}
 Since $B$ is the complement of $\Gamma_R$, Lemma~\ref{lem:correctableregionlemma} and the cleaning lemma applied to the region~$\Gamma_R$ imply that for any logical operator, we can find an equivalent logical operator with support completely contained in~$B$. This implies, in particular, that
 \begin{enumerate}[(i)]\setcounter{enumi}{2}
 \item
 there is a logical operator $\bar{P}_B$ with eigenvalues $\{+1,-1\}$ with support $\supp(\bar{P}_B)\subset B$.
 \end{enumerate}
 Consider two orthogonal encoded states $\bar{\rho}_0,\bar{\rho}_1$ with
\begin{align}
\tr(\bar{P}_B(\bar{\rho}_0-\bar{\rho}_1))=2\ .\label{eq:boundone}
 \end{align} 
Defining $\rho_b$ for $b\in \{0,1\}$ by $(\rho_b)_A\otimes \Phi_{A^c}:=U^\dagger \bar{\rho}_b U$ (where $U$ is a unitary encoder), we have 
\begin{align*}
\tr(\bar{P}_B(\bar{\rho}_0-\bar{\rho}_1))&=
\tr(\bar{P}_B U(\rho_0\otimes\Phi-\rho_1\otimes\Phi)U^\dagger)\ .
 \end{align*} 
 In particular, if $\cE^{(t)}$ encodes into~$\cC$ with error~$\epsilon$, then this implies (with $\tr(AB)\leq \|A\|_\infty\cdot \|B\|_1$ and $\|P_B\|_\infty=1$) that 
 \begin{align}
 \big|\tr(\bar{P}_B(\bar{\rho}_0-\bar{\rho}_1))-\tr(\bar{P}_B \cE^{(t)}((\rho_0-\rho_1)\otimes\Phi)\big|\leq \epsilon\ .\label{eq:boundtwo}
 \end{align}
 Combining~\eqref{eq:boundone} and~\eqref{eq:boundtwo}, we get for $P(t):=(\cE^{(t)})^\dagger(\bar{P}_B)$
 \begin{align}
 \tr\left(P(t)((\rho_0-\rho_1)\otimes\Phi)\right)\geq 2-\epsilon\ .\label{eq:ptlowerbound}
 \end{align}
 Set $P_r(t):=(\cE^{(t)}_{B_r})^\dagger(\bar{P}_B)$, where $\cE^{(t)}_{B_r}$ is the localized evolution (according to the Lieb-Robinson property). We then have
   using~\eqref{it:upperboundbsize}, $\|\bar{P}_B\|_\infty=1$ and choosing $r=\frac{c}{3}d^{\frac{1}{D-1}}$
 \begin{align*}
 \|P(t)-P_{r}(t)\|_\infty \leq C n\exp(vt-\gamma \frac{c}{3}d^{\frac{1}{D-1}})\ , 
 \end{align*}
 for some nonnegative constants $C,v,\gamma$. This, combined with~\eqref{eq:ptlowerbound} and 
 $\|(\rho_0-\rho_1)\otimes\Phi\|_1\leq 2$ gives
 \begin{align}
\begin{split} \tr\left(P_{r}(t)((\rho_0-\rho_1)\otimes\Phi)\right)&\geq\\
&\hspace{-12ex} 2-\epsilon-2Cn\exp(vt-\gamma \frac{c}{3}d^{\frac{1}{D-1}})\ .
\end{split}\label{eq:plthree}
 \end{align}
But $P_{r}(t)$ is, by definition supported on 
the set $B(r):=\{x\ |\ d(x,B)\leq r\}$. By definition of $r$ it is easy to check that
$B(r)\subset A^c$, i.e., it has no intersection with $A$. This implies
\begin{align}
 \tr\left(P_{r}(t)((\rho_0-\rho_1)\otimes\Phi)\right)=0\ .\label{eq:plfour}
\end{align}
   Eqs.~\eqref{eq:plthree} and~\eqref{eq:plfour} are compatible only if $t\geq \Omega(d^{\frac{1}{D-1}})$, as claimed.
   
   \end{proof} 

\section{Preparation of topological order}
Next we discuss the proof of Theorem~\ref{thm:preparationtoporder}. In contrast to the case of encoding maps, there is no distinguished subset of qudits which we can use to argue. Instead, we will use the fact that there is a pair of measurements which yields correlated results in any ground state, but independent outcomes for product states. If the preparation map is locality-preserving, this leads to a contradiction as this property is preserved.

\subsection{Proof sketch for the~$2D$ toric code}
Consider for simplicity the case of the toric code. 
Suppose we have an initial product state $\ket{\Phi}=\bigotimes_j \ket{\Phi_j}$, which
is transformed into an~$\epsilon$-approximation~$\cE^{(t)}(\Phi)\approx \rho_{GS}$ of a ground state~$\rho_{GS}$.  We will focus on the anticommuting pair of logical operators $(\bar{X},\bar{Z})$ associated with the first encoded qubit. For the state~$\rho_{GS}$ we know that the expectation values satisfy
\begin{align}
\langle \bar{X}\rangle_\rho^2 +  \langle \bar{Z} \rangle_\rho^2 \leq 1\ .\label{eq:uncertaintyrelationgstoric}
\end{align}  Let us without loss of generality assume that 
\begin{align*}
\langle \bar{Z}\rangle_\rho^2 \leq 1/2\ . 
\end{align*}  We also have that $\langle \bar{Z}^2 \rangle_\rho=1$.

Let us now take two incarnations~$\bar{Z}^{(1)}$ and $\bar{Z}^{(2)}$ of $\bar{Z}$ such that their support is separated by a distance~$L/2$ (concretely, $\bar{Z}^{(1)}$ is supported on a vertical strip, whereas $\bar{Z}^{(2)}$ is its translate in the horizontal direction). Since the encoder prepares a state which is approximately a ground state, we have that 
\begin{align*}
\langle \bar{Z}^{(1)}\bar{Z}^{(2)} \rangle_{\cE^{(t)}(\Phi)} \geq 1-\epsilon\ .
\end{align*} However, going to the Heisenberg picture and using locality, i.e.,  a spacially truncated evolution operator (with support on two disjoint regions~$B=B_1\cup B_2$) shows that this expectation value is approximately equal to the expectation value of a product operator~$ (\cE^{(t)})^\dagger(\bar{Z}^{(1)}\bar{Z}^{(2)})\approx (\cE_{B_1}^{(t)})^\dagger(\bar{Z}^{(1)})\otimes (\cE_{B_2}^{(t)})^\dagger(\bar{Z}^{(2)})$
 on the initial product state~$\Phi$. 
 The expectation of the product is greater than $1-\epsilon$ whereas the expectation of each of the operators individually is smaller than $1/\sqrt{2}$. This implies that the observables must be correlated in the initial state, contradicting the hypothesis that $\Phi$ is a product state.

\subsection{General proof of Theorem~\ref{thm:preparationtoporder} and consequences}
The previous argument was very specific to the $2D$ toric code. In particular, it relies on the fact that there are logical Pauli operators supported on a strip. Clearly, similar arguments apply to e.g., the $3D$ and $4D$ toric codes. More generally, we can extend the proof to codes satisfying the assumptions of Theorem~\ref{thm:preparationtoporder}. 
\begin{proof}
By assumption, there is a logical observable~$\logicalL$ which is  uncertain for~$\rho$. 
Furthermore, there are two POVMs $\cM^{(1)}=\{M_\alpha^{(1)}\}_\alpha$ and $\cM^{(2)}=\{M_\alpha^{(2)}\}_\alpha$ realizing~$\logicalL$. 
The assumptions of Theorem~\ref{thm:preparationtoporder} further state that the supports $B_j=\supp(\cM^{(j)})$ of the POVMs  $\cM^{(j)}$ for $j=1,2$ are separated by a distance \begin{align}
d(B_1,B_2)\geq c L \label{eq:distanceref}
\end{align}
 for some constant~$c$.   
Intuitively, the POVMs $\cM^{(1)}$ and $\cM^{(2)}$ constitute the counterpart of the logical observables~$\bar{Z}^{(1)}$ and $\bar{Z}^{(2)}$ used in the case of the toric code.
 
Consider  the measurement~$\cM=\cM^{(1)}\otimes\cM^{(2)}$ whose POVM elements are tensor products~$M^{(1)}_\alpha\otimes M^{(2)}_\beta$ with associated pairs $(\alpha,\beta)$ as outcomes. This measurement has support on~$B=B_1\cup B_2$. Furthermore, it yields perfectly correlated results  when measuring a ground state~$\rho$. In particular, the Shannon entropy of the measurement outcome~$\cM(\rho)$  is given by
\begin{align}
H(\cM(\rho))&=H(\cM^{(1)}(\rho))=H(\cM^{(2)}(\rho))\ .\label{eq:measuremententropydepend}
\end{align}
To characterize the outcome~$\cM(\cE^{(t)}(\Phi))$, we go to the Heisenberg picture according to
\begin{align*}
\tr[(M^{(1)}_\alpha\otimes M^{(2)}_\beta)\cE^{(t)}(\Phi)]&=\tr[\cE^{(t)\dagger} (M^{(1)}_\alpha\otimes M^{(2)}_\beta)\Phi].
\end{align*}
Because~$\cE^{(t)}$ can be localized to an $r$-neighborhood of~$B$, i.e., 
\begin{align}
\hspace{-2ex}\|(\cE^{(t)})^\dagger (M^{(1)}_\alpha\otimes M^{(2)}_\beta)-
(\cE^{(t)}_{B(r)})^\dagger (M^{(1)}_\alpha\otimes M^{(2)}_\beta)\|\nonumber\\
\qquad \leq |B|\cdot Ce^{vt-\gamma r}\label{eq:approximationerrorcEt}
\end{align}
we can approximately simulate the distribution over measurement outcomes~$\cM(\cE^{(t)}(\Phi))$ by considering a tensor product measurement
$\cM_{B_1(r)}^{(1)}\otimes\cM_{B_2(r)}^{(2)}$ applied to the initial state~$\Phi$ (here $\cM^{(j)}_{B_j(r)}=(\cE^{(t)}_{B_j(r)})^\dagger (\cM^{(j)})$).
Suppose the state $\cE^{(t)}(\Phi)$ is $\epsilon$-close to a ground state~$\rho$
for some $t \leq \frac{\gamma c L}{4 v}$, 
where  $\Phi$ is an initial product state (recall that
$(\gamma, v)$ are Lieb-Robinson-parameters, whereas $c$ is defined by~\eqref{eq:distanceref}).  Choosing $r= \frac{1}{2} c L$ and recalling that $|B| \in O(L^D)$ the approximation error~\eqref{eq:approximationerrorcEt} is an exponentially decaying function of~$L$.
 However, since $\Phi$ is a product state, the latter measurement gives uncorrelated measurement results,  implying 
\begin{align}
H(\cM(\cE^{(t)}(\Phi)))\approx \sum_{j=1,2}H(\cM^{(j)}_{B_j(r)}(\Phi))\ \label{eq:firstidentity}
\end{align}
(The precise meaning~$\approx$ is irrelevant here. It depends on continuity bounds for the von Neumann entropy.)

On the other hand, because $\cE^{(t)}(\Phi)$ and $\rho$ are $\epsilon$-close, we must have (by identical reasoning)
\begin{align*}
H(\cM(\rho))&\approx H(\cM(\cE^{(t)}(\Phi)))\quad\textrm{ and }\\
H(\cM^{(j)}(\rho))&\approx H(\cM_{B_j(r)}^{(j)}(\Phi))\qquad \textrm{ for each }j=1,2\ .
\end{align*}
Combining this with~\eqref{eq:measuremententropydepend} gives
\begin{align}
H(\cM(\cE^{(t)}(\Phi)))&\approx H(\cM^{(j)}_{B_j(r)}(\Phi))\ .\label{eq:secondidentity}
\end{align}
Taken together, Eq.~\eqref{eq:firstidentity} and~\eqref{eq:secondidentity} contradict the fact that the observable is uncertain. 
\end{proof}
Theorem~\ref{thm:preparationtoporder}  distills the underlying information-theoretic properties of logical operators. As an example of its application, consider the class of STS codes~\cite{Yoshida2011} in $D\leq 3$ dimensions. Two facts imply the existence of distant incarnations~$\cM^{(1)}, \cM^{(2)}$ of an uncertain logical observable~$\logicalL$ in this case: The first is  translational equivalence~\cite[Theorem 2]{Yoshida2011}. This states that any translation of a logical observable  with respect to a coarse grained lattice vector yields an equivalent logical observable.
The second ingredient is provided by  dimensional duality~\cite[Theorems 4 and 5]{Yoshida2011}: this implies implies  that in $D\leq 3$ dimensions, there exist canonical sets of logical operators $\{\ell_1, \ldots, \ell_k \}$ and $\{r_1, \ldots, r_k\}$ with canonical commutation relations and geometric dimensions adding pairwise to $D$.
Hence, for any encoded state $\rho$ there should be a constant amount of uncertainty associated to at least one observable from each pair (cf.~\eqref{eq:uncertaintyrelationgstoric}).

It remains to prove Corollary~\ref{cor:topologicalcodes}, which again boils down to identifying a suitable uncertain logical measurement.  Here we proceed similarly as in the case of the toric code, using a general uncertainty relation instead of~\eqref{eq:uncertaintyrelationgstoric}. 

Recall that the ground space of a system described by a TQFT on a torus  has two distinguished orthonormal bases~$\cB_1=\{\ket{i}_1\}_{i\in\cA}$ and $\cB_2=\{\ket{j}_2\}_{j\in\cA}$, where basis elements are indexed by anyon labels of the model. The bases~$\cB_1,\cB_2$ are related by the S-matrix, whose matrix elements, are, in the diagrammatic calculus of category theory, given by the Hopf link
\begin{align*}
S_{ij} = \frac{1}{\cD}\ \Smatrix{i}{j}\ .
\end{align*}
Here $\cD$ is the total quantum dimension. 

The index~$i$ (or `particle type') of the basis element~$\ket{i}_1$ can
be retrieved by measuring a logical observable ('string-operator') $\bar{L}_1$ 
which is supported on a strip~$C_1$ along a fixed topologically non-trivial cycle (or any of its translates).
 For states expressed in the second basis, there is a corresponding observable~$\bar{L}_2$ supported on a strip~$C_2$ along the  complementary non-trivial cycle.
\begin{proof}[Proof of Corollary~\ref{cor:topologicalcodes}]
The von Neumann measurements
corresponding to the eigenbases of $\bar{L}_1$ and $\bar{L}_2$ have realizations $\cM_{\bar{L}_1}$ and $\cM_{\bar{L}_2}$ supported on~$C_1$ and $C_2$, respectively.
The general uncertainty relation~\cite{MaassenUffink88} (corresponding to measuring in two complementary bases) tells us that for any state~$\rho$ supported on the ground space, we have 
\begin{align*}
H(\cM_{\bar{L}_1}(\rho))+H(\cM_{\bar{L}_2}(\rho))\geq -\log \max_{i,j} |S_{ij}|^2\ 
\end{align*}
In particular, this implies that there is an index $j\in \{1,2\}$ such that the measurement outcome~$\cM_{\bar{L}_j}(\bar{\rho})$ when measuring~$\rho$ is not deterministic: its Shannon entropy is lower bounded by a constant independent of the system size. 
 Using a $\cM^{(1)}\equiv\cM_{\bar{L}_j}$ and a measurement~$\cM^{(2)}$ obtained by translating $\cM^{(1)}$ by a distance~$L/2$ shows that the conditions of Theorem~\ref{thm:preparationtoporder} are satisfied. 
\end{proof}

In summary, our work establishes fundamental new limits on the preparation of topologically ordered states by possibly non-unitary  local processes. These limits are based on information-theoretic properties such as the underlying code distance and the structure of the logical operators. One may view these results in the general
context of classifying different phases: here the notion of local unitary circuits plays a crucial role in defining equivalence~\cite{ChenGuWen10}. By going beyond this restricted notion of equivalence, our work underscores the distinction
between topologically ordered and trivial phases. 

\subsection*{Acknowledgment}
RK gratefully acknowledges support  by NSERC  and thanks the Isaac Newton Institute for their hospitality.  FP acknowledges funding provided by
the Institute for Quantum Information and Matter, a
NSF Physics Frontiers Center with support of the Gordon
and Betty Moore Foundation (Grants No. PHY-0803371
and PHY-1125565).


\end{document}